\documentclass[12pt]{extarticle}
\usepackage[a4paper, total={7in, 9in}]{geometry}
\usepackage{amsmath,amsfonts}
\usepackage{algorithmic}
\usepackage{algorithm}
\usepackage{array}
\usepackage[caption=false,font=normalsize,labelfont=sf,textfont=sf]{subfig}
\usepackage{textcomp}
\usepackage{stfloats}
\usepackage{url}
\usepackage{verbatim}
\usepackage{graphicx}
\usepackage{cite}
\usepackage{authblk}

\usepackage{amsthm}

\newtheorem{theorem}{Theorem}[section]
\newtheorem{lemma}[theorem]{Lemma}
\newtheorem{definition}[theorem]{Definition}
\newtheorem{prop}[theorem]{Proposition}

\newtheorem{cor}[theorem]{Corollary}

\newtheorem{question}{Question}

\DeclareMathOperator{\KL}{KL}

\begin{document}

\title{Cryptographic Compression}

\author{Joshua Cooper and Grant Fickes \\ Atombeam Technologies}
\date{\today}

\maketitle

\begin{abstract}
We introduce a protocol called ENCORE which simultaneously compresses and encrypts data in a one-pass process that can be implemented efficiently and possesses a number of desirable features as a streaming encoder/decoder.  Motivated by the observation that both lossless compression and encryption consist of performing an invertible transformation whose output is close to a uniform distribution over bit streams, we show that these can be done simultaneously, at least for ``typical'' data with a stable distribution, i.e., approximated reasonably well by the output of a Markov model.  The strategy is to transform the data into a dyadic distribution whose Huffman encoding is close to uniform, and then store the transformations made to said data in a compressed secondary stream interwoven into the first with a user-defined encryption protocol.  The result is an encoding which we show exhibits a modified version of Yao's ``next-bit test'' while requiring many fewer bits of entropy than standard encryption.  Numerous open questions remain, particularly regarding results that we suspect can be strengthened considerably.
\end{abstract}

\section{Introduction}

Lossless compression and encryption share a goal: transforming strings reversibly and efficiently into an approximately uniformly random string.  Encryption seeks this end because there is no information that can be extracted from a uniform sequence, so the user's plaintext remains private; compression seeks this end because only sequences whose distribution is uniform cannot be further compressed, so the user's data is maximally shortened.  Naturally, the measures of success are different between these two perspectives.  Nonetheless, here we show that, at least for typical-looking streaming data -- that sampled from a Markov chain -- it is possible to achieve both types of benefit simultaneously, with a protocol we term ENCORE: ENcryption and COmpression REconciled.  The use of Markov models as a proxy for ``typicality'' is a standard going back at least to Shannon's foundational 1948 paper \cite{Shannon1948}.

We describe the ENCORE protocol in the next section.  Broadly, the strategy is comprised of two components.  First, the underlying distribution of data is transformed into a distribution that is {\em dyadic}, i.e., wherein each probability is a negative power of $2$.  As we show below in Section \ref{Sec:AlgorithmSecurity}, data which is sampled iid from a dyadic distribution gives rise to a Huffman encoding which is a pure uniform bit stream, i.e., it is indistinguishable from a sequence of iid fair coin tosses.  (It has been previously observed that Huffman coding provides a semblance of encryption; see, for example, \cite{GMR96}.)  We then show that data generated by a Markov model and then transformed in a particular manner has this property asymptotically, i.e., following a burn-in period whose length is governed by the mixing time of the model, bits of the Huffman encoded data are exponentially close to a fair coin toss -- even when conditioned on the entire past of the stream.

The second component of the protocol provides the ``reconstruction'' information to invert the transformations in order to recover the original data.  We want to ensure that few bits are used to perform the transformations -- as true random samples are notoriously costly.  We also aim to minimize the damage done by additional reconstruction data to the abridgement provided by Huffman coding.  Therefore, transformations are chosen to occur infrequently -- in practice, a greedy approach works well -- and the resulting sequence of modifications is compressed in a manner designed to take advantage of this sparseness.  This compressed transformation data can then be encrypted with any cryptographic scheme.  One might use two channels, if available, to send the encoded-transformed stream and the transformation-reconstruction stream.  However, for the purpose of self-containment of the protocol, we introduce in Section \ref{Section:Interleaving} a way to interleave the two streams into one which we prove avoids exposing information about the original data.  Finally, in Section \ref{Section:Questions}, we present a number of open questions about the protocol -- in particular, we conjecture that a number of the properties we prove it exhibits can be strengthened.

\section{The ENCORE Algorithm}\label{Sec:AlgorithmDescription}

Here we present the ENCORE (ENcrpytion and COmpression REconciled) algorithm whose cryptographic properties are formally considered in the sequel. We denote by $\Omega$ the set of states (typically, words of a fixed length over a finite alphabet generated by an ergodic Markov model) and write $\sigma$ for the stationary distribution of these states. The distribution $\sigma$ is used to generate a Huffman encoding, $C$, of $\Omega$. The encoding given by $C$ defines another ``implied'' distribution on $\Omega$, called $\pi$, where for each $\omega \in \Omega$, $\pi(\omega) := 2^{-|C(\omega)|}$, where $C(\omega)$ denotes the codeword assigned to $\omega$ by $C$ and $|\omega'|$ is the length of the word $\omega'$ in bits. 

States of $\Omega$ can be partitioned by the relation between $\sigma(\omega)$ and $\pi(\omega)$ as follows. Let $\mathcal{O} := \{\omega \in\Omega : \sigma(\omega) \geq \pi(\omega)\}$ be the collection of overrepresented (including perfectly represented) states, and let $\mathcal{U} := \{\omega \in\Omega : \sigma(\omega) < \pi(\omega)\}$ be the underrepresented states. From here, we define a nonnegative matrix $B = \{b_{\omega \omega'}\}_{\omega, \omega' \in \Omega}$ of $\Omega \mapsto \Omega$ ``transformation probabilities'' with the properties that:
\begin{enumerate}
    \item $B$ is row-stochastic: For all $\omega \in \Omega$, $\sum_{\omega' \in \Omega} b_{\omega \omega'} = 1$.
    \item If $\Omega$ is sampled according to $\sigma$ and then the elements are mapped to $\Omega$ according to the transformation probabilities $B$, the result is $\pi$: For all $\omega \in \Omega$, $\sum_{\omega'\in\Omega} \sigma(\omega') b_{\omega \omega'} = \pi(\omega)$.
    \item Underrepresented states only transform to themselves: for $\omega \in \mathcal{U}$, $b_{\omega \omega} = 1$.
    \item Overrepresented states only transform to (themselves and) underrepresented states: for $\omega \in \mathcal{O}$, $b_{\omega \omega'} > 0$ implies $\omega' \in \mathcal{U} \cup \{\omega\}$.
\end{enumerate} 
To see that it is always possible to define such a $B$, note that the following satisfies the above properties (where the third case can be ignored if $\mathcal{U} = \emptyset$, i.e., $\sigma$ is already dyadic and $\pi \equiv \sigma$): 
$$
b_{\omega \omega'} = 
\begin{cases}
    1 & \text{if $\omega = \omega' \in \mathcal{U}$} \\
    \pi(\omega)/\sigma(\omega) & \text{if $\omega = \omega' \in \mathcal{O}$} \\
    \displaystyle \left (1- \frac{\pi(\omega)}{\sigma(\omega)} \right ) \frac{\pi(\omega')-\sigma(\omega')}{\sum_{\alpha \in U} \pi(\alpha)-\sigma(\alpha)} & \text{if $\omega \in \mathcal{O}$ and $\omega' \in \mathcal{U}$} \\
    0 & \text{otherwise.}
\end{cases}
$$
In practice, it is preferable for row $\omega$ of $B$ to represent a probability distribution with small support whenever possible, so often we allocate the probabilities $\pi(\omega')-\sigma(\omega')$, $\omega' \in \mathcal{U}$, greedily to states $\omega \in \mathcal{O}$ to saturate the probabilities $1-\pi(\omega)/\sigma(\omega)$ of nontrivial transformation instead of proportionally across all elements of $\mathcal{U}$.  This choice does not impact the analysis below, but may offer some advantages in computational efficiency and minimal use of additional entropy.  (See Section \ref{Section:Questions} for more explanation.) As we show below (see Proposition \ref{Prop:MarkovChainBitProb}), Huffman encoding of $\pi$ -- which we refer to as the ``encoded-transformed data'' stream -- approaches a uniform distribution over bit streams exponentially quickly, and (see Section \ref{Section:Interleaving}) the data required to reconstruct $\omega$ from $\omega'$ is very limited.  This data can therefore be compressed effectively (in a manner that exploits its sparsity) and then encrypted by any mechanism.  The result of said compression and encryption we will refer to as the ``transformation-reconstruction data'' stream.  Thus, by carefully interleaving the two resulting types of information, we can attain one stream which is simultaneously shorter on average than the original source stream sampled from $\sigma$ and passes a modified next-bit test (q.v.~\cite{Yao82}) that says nearby bits cannot be predicted with substantial accuracy, even given all previous data.  Overall, the process described herein uses significantly less computational power and bits of entropy (to feed the encryption of transformation data) than is required for the original stream -- and the result can be locally decoded, i.e., decompression and decryption require only a short window of data surrounding the desired information in the encoded stream.

\section{Security of ENCORE Algorithm}\label{Sec:AlgorithmSecurity}

Throughout the rest of the manuscript we let $\Omega$ be a finite set with $|\Omega| \geq 2$. Define $MC$ to be an ergodic Markov Chain on state space $\Omega$. Let $\tau$ be the mixing time of $MC$ and $\sigma$ its stationary distribution. Let $C$ be the Huffman encoding of $\sigma$, with $\pi$ the dyadic distribution implied by $C$, i.e., for each $\omega \in \Omega$, $\Pr_\pi[\omega] = 2^{-|C(\omega)|}$ where recall that $|\omega'|$ is the length of the word $\omega'$. Define $T(C)$ to be the Huffman tree corresponding to $C$, and let $d(v)$ and $h(v)$ denote the depth (distance from the root) and height (distance between $v$ and the deepest ``descendant'') of vertex $v$ in the rooted tree $T(C)$, respectively. Note that the Huffman tree $T(C)$ is used to generate the encoding $C$, namely, state $\omega \in \Omega$ is realized as a leaf in $T(C)$, and the sequence of left-child traversals (yielding the bit ``0'') and right-child traversals (yielding the bit ``1'') on the unique path from the root of $T(C)$ to the leaf corresponding to $\omega$ is its encoding. For vertex $v \in T(C)$, we let $\mathcal{E}_v \subseteq \Omega$ denote the event that $v$ is traversed in the encoding of state $\omega \in \Omega$. 

Let $M$ and $m$ be the lengths of the maximum and minimum length codeword given by $C$. Define $\mathcal{O}$ and $\mathcal{U}$ to be the overrepresented and underrepresented states in $\Omega$, i.e., $\mathcal{O} = \{\omega \in \Omega : \sigma(\omega) \geq \pi(\omega)\}$ and $\mathcal{U} = \{\omega \in \Omega : \sigma(\omega) < \pi(\omega)\}$. Note that perfectly-represented states, i.e., states $\omega \in \Omega$ satisfying $\sigma(\omega) = \pi(\omega)$, are contained in $\mathcal{O}$.  (This choice is an inconsequential convenience.)  If $\mathbf{x} = \{x_i\}_{i=1}^\infty$ is a sequence of states from $\Omega$ for all positive integer times $i$, we use $C(\mathbf{x})$ to denote the infinite word resulting from the concatenation of $C(x_i)$ for all $i \geq 1$. Furthermore, we let $B$ be the square matrix of dimension $|\Omega| \times |\Omega|$ with rows and columns indexed by $\Omega$ where entry $B_{i,j}$ contains the probability that state $i$ is transformed to state $j$ by the ENCORE algorithm. Thus, $B$ is row-stochastic, i.e., the sum of entries in row $i$ of $B$ is one for any $i \in \Omega$, and the entries of $B$ are all nonnegative. Lastly, let $C(x)_i$ denote the $i$-th bit of $C(x)$ for any $1\leq i\leq |C(x)|$. 

Note that while ``dyadic rational'' sometimes refers to any fraction whose denominator is an integral power of two, when we refer to a distribution $\sigma$ on state space $\Omega$ as ``dyadic'', we mean that $\sigma(\omega) = 1 / 2^k$ for some $k = k(\omega) \in \mathbb{N}$ for every $\omega \in \Omega$. Crucially, $\pi$ is always dyadic. Moreover, if $\sigma$ is dyadic, then $\sigma = \pi$. This observation is simple enough, so we omit its proof.  

\begin{lemma}\label{Lem:VertexDepthProb}
For every $v \in V(T(C))$, $\Pr_\pi[\mathcal{E}_v] = 2^{-d(v)}$. 
\end{lemma}
\begin{proof}
Let $v \in V(T(C))$. Let $T'$ denote the subtree of $T(C)$ rooted at $v$ induced by ($v$ and) all descendants of $v$. If $L$ is the set of all leaves in $T'$, then 
$$
\Pr_\pi[\mathcal{E}_v] = \sum_{\ell \in L} \pi(\ell). 
$$
For every leaf $\ell \in V(T(C))$, $\Pr_\pi[\mathcal{E}_\ell] = \pi(\ell) = 2^{-d(\ell)}$. Thus, 
$$
\Pr_\pi[\mathcal{E}_v] = \sum_{\ell \in L} 2^{-d(\ell)} = 2^{-d(v)} \sum_{\ell \in L} 2^{d(v) -d(\ell)} = 2^{-d(v)} \sum_{\ell \in L} \Pr_\pi[\mathcal{E}_\ell | \mathcal{E}_v] =  2^{-d(v)}.
$$ 
\end{proof}

\begin{cor}\label{Cor:EqualBranchProbs}
Let $v$ be a non-leaf vertex in $T(C)$. If $v_0$ and $v_1$ are the left and right children of $v$, then $\Pr_\pi[\mathcal{E}_{v_0}] = \Pr_\pi[\mathcal{E}_{v_1}] = \Pr_\pi[\mathcal{E}_v] / 2$. 
\end{cor}

\begin{theorem}\label{Thm:iidDyadicSourceSecure}
If $\mathbf{x} = \{x_k\}_{k=1}^\infty$ is a sequence of elements drawn i.i.d.~from $\pi$, then $\Pr_\pi[C(\mathbf{x})_j = 0 | C(\mathbf{x})_1, \ldots, C(\mathbf{x})_{j-1}] = 1/2$ for all $j \geq 1$. 
\end{theorem}
\begin{proof}
Let $\mathbf{x} = \{x_k\}_{k=1}^\infty$ be an i.i.d.~sampling from the dyadic distribution $\pi$. 
Define $\{y_l\}_{l=0}^\infty$ so that $y_0 = 0$ and $y_{l} = y_{l-1} + |C(x_{l})|$ for each $l\geq 1$. Fix any $i > 0$, and let $k \in \mathbb{N}$ so that $k = \min\{l : y_l \geq i\}$. Thus, the bit $C(\mathbf{x})_i$ is generated by $C(x_k)$. We split into two cases depending on whether $C(\mathbf{x})_i$ is the first bit of $C(x_k)$. 

Case $1$: The bit $C(\mathbf{x})_i$ is the first bit of $C(x_k)$. Since the unique vertex $v$ of $T(C)$ at depth zero is not a leaf, by Corollary \ref{Cor:EqualBranchProbs}, if $v_0$ and $v_1$ are the left and right children of $v$, then $\Pr_\pi[C(\mathbf{x})_i = 0] = \Pr_\pi[\mathcal{E}_{v_0}] = 1/2$. Since the sample $\mathbf{x}$ is i.i.d., the value $C(\mathbf{x})_i$ is independent of $\{C(\mathbf{x})_j\}_{j=1}^{i}$ for each $1 \leq j < i$. 

Case $2$: The bit $C(\mathbf{x})_i$ is not the first bit of $C(x_k)$. By the definition of $k$, note that $C(\mathbf{x})_{y_{k-1}+1}$ is the first bit of $C(\mathbf{x})$ generated by $C(x_k)$. Since the sample $\mathbf{x}$ is i.i.d., $\Pr_\pi[C(\mathbf{x})_i | C(\mathbf{x})_1, \ldots, C(\mathbf{x})_{i-1}] = \Pr_\pi[C(\mathbf{x})_i | C(\mathbf{x})_{y_{k-1}+1}, \ldots, C(\mathbf{x})_{i-1}]$. Furthermore, if $v \in V(T(C))$ is the parent vertex of the edge which generates the bit $C(\mathbf{x})_i$ (i.e., the vertex of depth $i-y_{k-1}-1$ in $T(C)$ on the path from the root to $x_k$), then $\Pr_\pi[C(\mathbf{x})_i | C(\mathbf{x})_{y_{k-1}+1}, \ldots, C(\mathbf{x})_{i-1}] = \Pr_\pi[C(\mathbf{x})_i | \mathcal{E}_v]$, but by Corollary \ref{Cor:EqualBranchProbs}, $\Pr_\pi[C(\mathbf{x})_i | \mathcal{E}_v] = 1/2$, completing the proof. 
\end{proof}

For two distributions $D_1$ and $D_2$ over the same countable collection of symbols $\Omega$, the total variation between the two distributions, $\|D_1 - D_2\|$ is defined by 
$$
\|D_1 - D_2\| = \sum_{\omega \in \Omega} \frac12 |D_1(\omega) - D_2(\omega)| = \max_{A\subset \Omega} \left(\Pr_{D_1}[A] - \Pr_{D_2}[A]\right) .
$$
Further, for Markov chain $MC$ with stationary distribution $\sigma$, let $\Delta(t) := \max_\omega \|\sigma - P_\omega^t\|$, where $t$ is some positive integer and $P_\omega^t$ denotes the distribution $t$ steps after state $\omega$. Thus, the mixing time of Markov chain $MC$ is defined by 
$$
\tau = \min\left\{ t : \Delta(t) \leq \frac{1}{2e} \right\} .
$$
There is not consensus on the choice $1/2e$. Some places in the literature use $1/4$ instead. We adopt the value $1/2e$, since it yields the following result. 

\begin{prop}\label{Prop:MixTimeThreshold}
(pg 55 in \cite{LevinPeresWilmer2006})
For any finite ergodic Markov chain, $\Delta(t)$ is nonincreasing. Additionally, 
for positive time $n$, 
$$
\Delta(n) \leq \frac{1}{e^{n / \tau}}
$$
\end{prop}

\begin{definition}
Let $X$ and $Y$ be random variables with probability distributions $\mu$ and $\nu$ on $\Omega$. A distribution $\rho$ on $\Omega\times\Omega$ is a coupling if 
\begin{align*}
\forall x \in \Omega, \qquad & \sum_{y \in \Omega} \rho(x,y) = \mu(x) \\
\forall y \in \Omega, \qquad & \sum_{x \in \Omega} \rho(x,y) = \nu(y) .
\end{align*}
\end{definition}

\begin{lemma}\label{Lem:Couple}
(Proposition 4.7 in \cite{LevinPeresWilmer2006})
Consider a pair of distributions $\mu$ and $\nu$ over $\Omega$. 
\begin{enumerate}
\item For any coupling $\rho$ of $\mu$ and $\nu$, let $(X,Y)$ be sampled from $\rho$.  Then
$$
\|\mu - \nu\| \leq \Pr(X\neq Y) .
$$
\item There always exists a coupling $\rho$ so that 
$$
\|\mu - \nu\| = \Pr(X\neq Y) .
$$
\end{enumerate}
\end{lemma}
\begin{proof}
By the definition of a coupling, $\rho(\omega ,\omega) \leq \min\{\mu(\omega), \nu(\omega)\}$ for all $\omega \in \Omega$. Therefore, 
\begin{align*}
\Pr(X \neq Y) &= 1 - P(X = Y) \\ 
&= 1 - \sum_{\omega} \rho(\omega, \omega) \\
&\geq \sum_{\omega} (\mu(\omega) - \min\{\mu(\omega), \nu(\omega)\}) \\
&\geq \sum_{\omega: \mu(\omega) > \nu(\omega)} \mu(\omega) - \nu(\omega) \\ 
&= ||\mu - \nu||.
\end{align*}

For (2), define coupling by $\rho(z,z) = \min\{\mu(z), \nu(z)\}$ for all $z$ and for distinct $y,z$, 
$$
\rho(y,z) = \frac{(\mu(y) - \rho(y,y))(\nu(z) - \rho(z,z))}{1 - \sum_x \rho(x,x)} .
$$
\end{proof}

\begin{prop}\label{Prop:MarkovChainBitProb} 
For any positive integer $n$, let $P^{n}$ denote the distribution of states after $n$ steps. Let $\mathbf{x}$ be a sampling for all positive times from $MC$. If $z_i$ denotes the $i$-th bit of binary string $z$, then for any $j > n(2M-m)$, 
$$
\left |\Pr[C(\mathbf{x})_j = 0] -\frac12 \right | \leq 2\cdot \frac{\exp[-\lfloor j / (2M-m) \rfloor/\tau]}{1 - e^{-1/\tau}} .
$$
\end{prop}
\begin{proof}
By Proposition \ref{Prop:MixTimeThreshold}, 
$\|P^n -\pi\| \leq e^{-n/\tau}$. Let the sequence of random variables $\{X_k\}_{k=1}^\infty$ be distributed as samples from MC (i.e., $X_k \sim P^k|X_{k-1}$ for all $k \geq 2$) and the sequence of random variables $\{Y_k\}_{k=1}^\infty$ be i.i.d.~sampled according to $\pi$. Then $\| X_n - Y_n\| \leq e^{-n/\tau}$, so Lemma \ref{Lem:Couple} gives the existence of coupling $\rho_n$ so that $\Pr(X_n\neq Y_n) \leq e^{-n/\tau}$. Let $(x_n,y_n)$ be a sample from $\rho_n$ for each $n\geq 1$, so the marginals of $x_n$ and $y_n$ are $X_n$ and $Y_n$ respectively. Thus, $\Pr(x_n\neq y_n) \leq e^{-n/\tau}$.  Define $\Delta_N := |C(x_1\cdots x_N)| - |C(y_1\cdots y_N)|$. Then 
\begin{align*}
\mathbb{E} \left [ \Delta_N \right ] &\leq \sum_{i=1}^N (M-m) \cdot \Pr[x_i \neq y_i] \\
&\leq \sum_{i=1}^\infty Me^{-i/\tau} \\
&\leq \frac{M}{1 - e^{-1/\tau}}
\end{align*}
for any constant $N\geq 1$. Notice $|\Delta_N|\leq (M-m)N$. 
Further, if $t$ is some positive integer, then define $E_{N,t}$ to be the event that $\{x_{N+i}\}_{i=0}^{t} \neq \{y_{N+i}\}_{i=0}^{t}$. Moreover, if $t = \infty$, we collapse the notation to $E_N$, so that 
$$
\Pr[E_{N,t}] < \Pr[E_N] \leq \sum_{i=0}^{\infty} \Pr[x_{N+i} \neq y_{N+i}] \leq \sum_{i=0}^{\infty} e^{-(N+i)/\tau} = \frac{e^{-N/\tau}}{1 - e^{-1/\tau}}.
$$
The previous computation implies that the probability the sequence $\{\Delta_i\}_{i=N}^\infty$ is not constant is at most $e^{-N/\tau} / (1 - e^{-1/\tau})$. Choose $t$ so that $t \geq 3(N+1)((M/m) - 1) + 1$, and suppose $j \geq N(2M-m)$. Then if $N$ is fixed so that $N = \lfloor j / (2M-m) \rfloor$, then 
\begin{align*}
N(2M-m) \leq j & \leq (N+1)(2M-m)\\
& = (N+1)(2m-M) + (N+1)(3M-3m) \\
& = (N+1)(2m-M) + 3(N+1)(M/m-1)m \\
& \leq (N+1)(2m-M) + (t-1)m.
\end{align*}
Since $NM \leq N(2M-m)$ and $(N+1)(2m-M) + (t-1)m \leq (N+t)m$, the $j$-th bit of $C(x_1\cdots x_{N+t})$ falls within the encoding of $\{x_i\}_{i=N}^{N+t}$. Furthermore, since $|\Delta_N| \leq N(M-m)$, $j \geq N(2M-m)$ implies 
$$
j + \Delta_n \geq N(2M-m) - N(M-m) = NM,
$$ and $j\leq (N+1)(2m-M) + (t-1)m$ implies 
\begin{align*}
j + \Delta_N &\leq (N+1)(2m-M) + (t-1)m + N(M-m) \\
&= Nm + 2m -M + tm - m \\
&= (N + t)m + (m - M) \\
&\leq (N+t)m.
\end{align*}
Thus, the $(j+\Delta_N)$-th bit of $C(y_1\cdots y_{N+t})$ falls within the encoding of $\{y_i\}_{i=N}^{N+t}$. 
Writing $\mathbf{x} = \{x_i\}_{i=1}^\infty$ and  $\mathbf{y} = \{y_i\}_{i=1}^\infty$, we have
\begin{align} \label{eq1}
&\Pr[C(\mathbf{x})_j \neq C(\mathbf{y})_{j+\Delta_N}] \nonumber \\
&= \Pr[C(\mathbf{x})_j \neq C(\mathbf{y})_{j+\Delta_N} | E_{N,t}] \cdot \Pr[E_{N,t}] \nonumber \\
&\qquad + \Pr[C(\mathbf{x})_j \neq C(\mathbf{y})_{j+\Delta_N} | \overline{E}_{N,t}] \cdot \Pr[\overline{E}_{N,t}] \nonumber \\
&= \Pr[C(\mathbf{x})_j \neq C(\mathbf{y})_{j+\Delta_N} | E_{N,t}]\cdot \Pr[E_{N,t}] \nonumber \\
& \leq \Pr[E_{N,t}] \nonumber \\
& \leq \frac{e^{-N/\tau}}{1 - e^{-1/\tau}} = \frac{\exp[-\lfloor j / (2M-m) \rfloor/\tau]}{1 - e^{-1/\tau}}
\end{align}
Note that the second equality is due to the fact that the event $\overline{E}_{N,t}$ implies the two bits under consideration are the same, so $\Pr[C(\mathbf{x})_j \neq C(\mathbf{y})_{j+\Delta_N} | \overline{E}_{N,t}] = 0$. We continue by considering $\Pr[C(\mathbf{x})_j = 0]$. Partitioning this event according to whether we have equality of $C(\mathbf{x})_j$ and $C(\mathbf{y})_{j+\Delta_N}$ yields 
\begin{align*}
\Pr[C(\mathbf{x})_j = 0] &= \Pr[(C(\mathbf{x})_j = 0) \wedge (C(\mathbf{x}_j) = C(\mathbf{y})_{j+\Delta_N})]  \\
&\qquad + \Pr[(C(\mathbf{x})_j = 0) \wedge (C(\mathbf{x}_j) \neq C(\mathbf{y})_{j+\Delta_N})].
\end{align*}
Then (\ref{eq1}) can be applied to the second term on the right hand side, yielding 
\begin{align*}
\Pr[(C(\mathbf{x})_j = 0) \wedge (C(\mathbf{x})_j \neq C(\mathbf{y})_{j+\Delta_N})] &\leq \Pr[C(\mathbf{x})_j \neq C(\mathbf{y})_{j+\Delta_N}] \\
& \leq \frac{\exp[-\lfloor j / (2M-m) \rfloor/\tau]}{1 - e^{-1/\tau}} , 
\end{align*}
so 
$$
\Pr[C(\mathbf{x})_j = 0] \leq \Pr[(C(\mathbf{x})_j = 0) \wedge (C(\mathbf{x})_j = C(\mathbf{y})_{j+\Delta_N})] + \frac{\exp[-\lfloor j / (2M-m) \rfloor/\tau]}{1 - e^{-1/\tau}}.
$$
To bound the first term on the right side of this inequality, we partition the event based on $E_{N,t}$: 
\begin{align*}
&\Pr[(C(\mathbf{x})_j = 0) \wedge (C(\mathbf{x})_j = C(\mathbf{y})_{j+\Delta_N})] \\
&= \Pr[(C(\mathbf{x})_j = 0) \wedge (C(\mathbf{y})_{j+\Delta_N} = 0)] \\
&= \Pr[(C(\mathbf{x})_j = 0) \wedge (C(\mathbf{y})_{j+\Delta_N} = 0) | E_{N,t}]\cdot \Pr[E_{N,t}] \\
&\qquad + \Pr[(C(\mathbf{x})_j) = 0) \wedge (C(\mathbf{y})_{j+\Delta_N} = 0) | \overline{E}_{N,t}] \cdot \Pr[\overline{E}_{N,t}]
\end{align*}
The first term in the last line is bounded above by $\Pr[E_{N,t}]$, so 
\begin{align*}
&\Pr[(C(\mathbf{x})_j = 0) \wedge (C(\mathbf{x})_j = C(\mathbf{y})_{j+\Delta_N})] \\
&\leq \frac{\exp[-\lfloor j / (2M-m) \rfloor/\tau]}{1 - e^{-1/\tau}} + \Pr[(C(\mathbf{x})_j = 0) \wedge (C(\mathbf{y})_{j+\Delta_N}) = 0) | \overline{E}_{N,t}] \cdot \Pr[\overline{E}_{N,t}] .
\end{align*}
As for the second term, recall $\overline{E}_{N,t}$ is the event that $\{x_{N+i}\}_{i=0}^t = \{y_{N+i}\}_{i=0}^t$, in which case $C(\mathbf{x})_j = C(\mathbf{y})_{j+\Delta_N}$. Thus, 
$$
\Pr[(C(\mathbf{x})_j = 0) \wedge (C(\mathbf{y})_{j+\Delta_N} = 0) | \overline{E}_{N,t}] \cdot \Pr[\overline{E}_{N,t}] = \Pr[(C(\mathbf{y})_{j+\Delta_N} = 0) \wedge \overline{E}_{N,t}] .
$$
Since
$$
\Pr[(C(\mathbf{y})_{j+\Delta_N} = 0) \wedge \overline{E}_{N,t}] = \Pr[(C(\mathbf{y})_{j+\Delta_N} = 0)] - \Pr[(C(\mathbf{y})_{j+\Delta_N} = 0) \wedge E_{N,t}]
$$
and $\Pr[(C(\mathbf{y})_{j+\Delta_N} = 0)] = 1/2$ by Theorem \ref{Thm:iidDyadicSourceSecure} since $\mathbf{y}$ is sampled iid from the dyadic distribution $\pi$, 
$$
\Pr[(C(\mathbf{y})_{j+\Delta_N} = 0) \wedge \overline{E}_{N,t}] \leq \frac12_.
$$
Therefore, 
$$
\Pr[C(\mathbf{x})_j = 0] \leq \frac12 + 2\cdot \frac{\exp[-\lfloor j / (2M-m) \rfloor/\tau]}{1 - e^{-1/\tau}}
$$
for any integer $j\geq n(2M-m)$. Since the same computation holds for $\Pr[C(\mathbf{x})_j = 1]$, we have 
$$
\frac12 - 2\cdot \frac{\exp[-\lfloor j / (2M-m) \rfloor/\tau]}{1 - e^{-1/\tau}} \leq \Pr[(C(\mathbf{x})_j) = 0] \leq \frac12 + 2\cdot \frac{\exp[-\lfloor j / (2M-m) \rfloor/\tau]}{1 - e^{-1/\tau}}_,
$$
completing the proof.
\end{proof}

\section{Interleaving the Streams}\label{Section:Interleaving}

\begin{definition}
Let $X = \{x_i\}_{i=1}^\infty$ and $Y = \{y_i\}_{i=1}^\infty$ be infinite sequences. Let $S, T \subseteq \mathbb{Z}^+$ form a partition of the positive integers. Suppose $S = \{s_i\}_{i=1}^\infty$ and $T = \{t_i\}_{i=1}^\infty$ so that for any $i < j$, $s_i < s_j$ and $t_i < t_j$. Then the interleaving of $X$ and $Y$ with respect to the partition $(S,T)$ is the sequence $Z = \{z_i\}_{i=1}^\infty$ where $z_i = x_{s_j}$ if $s_j = i$ and $z_i = y_{t_j}$ if $t_j = i$.  That is, $Z$ is the order-preserving interleaving of $X$ and $Y$ with the property that the bits of $Z$ arising from $X$ are indexed by $S$ and the bits of $Z$ arising from $Y$ are indexed by $T$.
\end{definition}

\begin{lemma}\label{Lem:InterleaveUniform}
Let $X$ and $Y$ be two infinite uniform binary strings. Let $(S,T)$ be a random partition of the positive integers so that $Z = \{z_n\}_{n\geq 1}$ is the interleaving of $X$ and $Y$ with respect to partition $(S,T)$. Let $k_j := |\{s \in S : s \leq j\}|$, i.e., the number of bits in $Z$ with index at most $j$ which come from $X$. If the event $X_{k_j}=0$ and the event $Y_{j-k_j} = 0$ are (pairwise) independent of the event $j \in S$, then $\Pr[Z_j = 0] = 1/2$. 
\end{lemma}
\begin{proof}
We consider $\Pr[Z_j = 0]$. 
Let $k_j$ be as is defined in the statement of the lemma. Then, 
\begin{align*}
\Pr[Z_j = 0] &= \Pr[Z_j = 0 \wedge j \in S] + \Pr[Z_j = 0 \wedge j \in T] \\
&= \Pr[X_{k_j} = 0 \wedge j \in S] + \Pr[Y_{j - k_j} = 0 \wedge j \in T] \\
&= \Pr[X_{k_j} = 0] \cdot \Pr[j \in S] + \Pr[Y_{j - k_j} = 0] \cdot \Pr[j \in T]
\end{align*}
where the last equality is given by the assumption that the event $X_{k_j} = 0$ and the event $Y_{j-k_j} = 0$ are (pairwise) independent of the event $j \in S$. Thus, $\Pr[X_{k_j} = 0] = \Pr[Y_{j - k_j} = 0] = 1/2$ implies the desired result.
\end{proof}

\begin{prop}\label{Prop:InterleaveMaxBound}
Let $X = \{x_j\}_{j=1}^\infty$ and $Y = \{y_j\}_{j=1}^\infty$ be two infinite binary strings so there exists nonnegative sequences $\{a_j\}_{j=1}^\infty$ and $\{b_j\}_{j=1}^\infty$ satisfying 
$$
\left| \Pr[x_j = 0] - \frac12\right| \leq a_j \qquad \mbox{and} \qquad \left| \Pr[y_j = 0] - \frac12\right| \leq b_j
$$
for all $j\geq 1$. Further, let $(S,T)$ be a random partition of the positive integers, and let $Z$ be the binary word resulting from the interleaving of $X$ and $Y$ with respect to $(S,T)$. Let $k_j := |\{s \in S : s \leq j\}|$, i.e., the number of bits in $Z$ with index at most $j$ which come from $X$, and further, assume that the event $x_{k_j}=0$ and the event $y_{j-k_j} = 0$ are (pairwise) independent of the event $j \in S$.  Then for any integer $j \geq 1$, 
$$
\left| \Pr[z_j = 0] - \frac12\right| \leq \max\left\{ a_{k_j}, b_{j-k_j}\right\} .
$$
\end{prop}
\begin{proof}
We consider $\Pr[z_j = 0]$ for some $j\geq 1$. Clearly, since $(S,T)$ is a partition of the positive integers, 
$$
\Pr[z_j = 0] = \Pr[z_j = 0 \wedge j \in S] + \Pr[z_j = 0 \wedge j \in T].
$$
However, this is the same as 
$$
\Pr[z_j = 0] = \Pr[x_{k_j} = 0 \wedge j \in S] + \Pr[y_{j-k_j} = 0 \wedge j \in T].
$$
By the independence assumptions, we continue with the following computation. 
\begin{align*}
\Pr[z_j = 0] &= \Pr[x_{k_j} = 0]\Pr[j \in S] + \Pr[y_{j-k_j} = 0] \Pr[j \in T] \\
&\leq \left(\frac12 + a_{k_j}\right) \Pr[j \in S] + \left(\frac12 + b_{j-k_j}\right) \Pr[j \notin S] \\
&\leq \frac12 + \max\{a_{k_j}, b_{j-k_j}\} 
\end{align*}
The existence of the same upper bound for $\Pr[z_j = 1]$, completes the proof. 
\end{proof}

\begin{cor}\label{Cor:GeneralInterleaving}
Let $\mathbf{a} = \{a_j\}$ be a sampling for all positive times from $MC$, a finite ergodic Markov Chain with mixing time $\tau$. Let $X = \{x_j\}$ be the encoded-transformed data arising from $\mathbf{a}$ in the ENCORE algorithm, and $B$ the associated probability transformation matrix; let $Y = \{y_j\}$ be the transformation-reconstruction data generated in transforming $\mathbf{a}$; let $M$ denote the maximum length in bits of the image of the Huffman encoding $C$, and let $m$ denote the minimum such length. Suppose $Y$ is generated in such a way that the nonnegative sequence $\{b_j\}$ satisfies 
$$
\left| \Pr[y_j = 0] - \frac12\right| \leq b_j
$$
for all $j \geq 1$. Let $(S,T)$ be a partition of the positive integers so that for some positive integers $k$ and $l$, $S = \{i \in \mathbb{Z}^+ : 1 \leq i \mod(k+l) \leq k\}$, i.e., $S$ contains $k$ consecutive positive integers, then $T$ contains $l$ consecutive positive integers, etc. 
If $Z$ is the interleaving of $X$ and $Y$ with respect to the partition $(S,T)$, $n$ denotes some positive time, and 
$$
j' = k\cdot \left\lfloor \frac{j}{k+l} \right\rfloor + \begin{cases} j \mod(k+l) & \mbox{if } j \mod(k+l) \leq k \\ k & \mbox{otherwise,} \end{cases}
$$
denoting the number of bits of $\{z_i\}_{i=1}^j$ coming from $X$, then 
$$
\left| \Pr[z_j = 0] - \frac12\right| \leq \max\left\{2\cdot \frac{\exp[-\left\lfloor j' / (2M-m)\right\rfloor / (\tau \|B\|_1)]}{1 - \exp[-1 / (\tau\|B\|_1)]}, b_{j-j'} \right\}
$$
for any $j' > n \|B\|_1(2M-m)$. 
\end{cor}
\begin{proof} 
Let $A$ be the transition matrix for the Markov Chain $MC$, i.e., entry $a_{\omega\omega'}$ contains the probability that the next state is $\omega'$ if the current state is $\omega$, and recall that $B$ is the matrix where entry $b_{\omega\omega'}$ is the probability that $\omega$ is transformed to $\omega'$ by the ENCORE algorithm. Then $A$ and $B$ are both square of order $|\Omega| \times |\Omega|$, and if $\mathbf{e}_\omega$ denotes the elementary row vector in $\mathbb{R}^{\Omega}$ with a $1$ in the $\omega$ coordinate, then the vector $\mathbf{e}_\omega A$ is the distribution of the next state if the current state is $\omega$, and $\mathbf{e}_\omega B$ is the distribution of the transformed state if the current state is $\omega$. Define the operator $1$-norm of $B$ by 
$$
\|B\|_1 = \max_{\omega'} \sum_{\omega \in \Omega} |b_{\omega\omega'}|.
$$
Since the entries of $B$ are probabilities, $|b_{\omega\omega'}| = b_{\omega\omega'}$ for all entries, so $\|B\|_1$ is merely the largest column sum. Furthermore, $\tau \|B\|_1$ provides an upper bound on the mixing time of the transformed states (by sub-multiplicativity of the operator $1$-norm). Thus, by Proposition \ref{Prop:MarkovChainBitProb}, for any positive time $n$ and $j > n\|B\|_1(2M-m)$,
$$
\left| \Pr[x_j = 0] - \frac12 \right| \leq 2\cdot \frac{\exp[-\left\lfloor j' / (2M-m)\right\rfloor / (\tau \|B\|_1)]}{1 - \exp[-1 / (\tau\|B\|_1)]} .
$$
Since $\left| \Pr[y_j = 0] - \frac12\right| \leq b_j$, Proposition \ref{Prop:InterleaveMaxBound} implies 
$$
\left| \Pr[z_j = 0] - \frac12 \right| \leq \max\left\{2\cdot \frac{\exp[-\left\lfloor j' / (2M-m)\right\rfloor / (\tau \|B\|_1)]}{1 - \exp[-1 / (\tau\|B\|_1)]}, b_{j-j'} \right\} ,
$$
if $j'$ is the number of bits in $\{z_i\}_{i=1}^j$ taken from $X$ and $j' > n \|B\|_1(2M-m)$. Recall that $S$ and $T$ alternatingly contain $k$ and $l$ positive integers respectively, so 
$$
j' = k\cdot \left\lfloor \frac{j}{k+l} \right\rfloor + \begin{cases} j \mod(k+l) & \mbox{if } j \mod(k+l) \leq k \\ k & \mbox{otherwise,} \end{cases}
$$
completing the proof. 
\end{proof}

The next lemma allows us to use the preceding result without reference to the matrix $B$; see the concluding section below for an argument that it is likely possible to further improve these bounds.

\begin{lemma}\label{Lem:one-normBbound}
Let $B = \{b_{\omega\omega'}\}_{\omega,\omega' \in \Omega}$ be the transformation matrix defined by the ENCORE algorithm, i.e., $b_{\omega\omega'}$ denotes the probability that state $\omega$ is transformed to state $\omega'$. Then 
$$
\|B\|_1 \leq |\mathcal{O}| + 1 - \sum_{\omega' \in \mathcal{O}} \frac{\pi(\omega')}{\sigma(\omega')} \leq |\Omega|.
$$
\end{lemma}
\begin{proof}
If $\omega' \in \mathcal{O}$, then the $\omega'$ column of $B$ contains a value $0<v\leq 1$ in the diagonal entry and zeros elsewhere. In other words, if $\omega' \in \mathcal{O}$, then ENCORE does not transform any state of $\Omega$ to a different state in $\mathcal{O}$. Thus, it only remains to consider if $\omega' \in \mathcal{U}$. In this case, the diagonal entry of the $\omega'$ column of $B$ contains the value $1$, however, the remaining entries are not all zero. For the state $\omega'$, we know the following. 
$$
\sum_{\omega \in \Omega} \sigma(\omega) b_{\omega\omega'} = \pi(\omega') \qquad \Rightarrow \qquad \sum_{\omega\neq \omega'} \sigma(\omega) b_{\omega\omega'} = \pi(\omega') - \sigma(\omega')
$$
Since $B$ is row-stochastic and square of size $|\Omega| \times |\Omega|$, the sum of all entries of $B$ is $|\Omega|$. Thus, $\|B\|_1 \leq |\Omega|$. Furthermore, since the diagonal entry is $1$ for the columns corresponding to states of $\mathcal{\mathcal{U}}$, $\|B\|_1 \leq |\Omega| - |\mathcal{U}| + 1 = |\mathcal{O}| + 1$, since every entry of $B$ is at most one. The diagonal entry of columns corresponding to $\omega' \in \mathcal{O}$ is $\pi(\omega') / \sigma(\omega')$. Thus, 
$$
\|B\|_1 \leq |\Omega| + 1 - \left(|\mathcal{U}| + \sum_{\omega' \in \mathcal{O}} \frac{\pi(\omega')}{\sigma(\omega')} \right),
$$
where $|\Omega|$ is the sum of all entries of $B$, $1$ is the maximum value of any diagonal entry of $B$, and the term in parentheses subtracted off is the sum of all diagonal entries of $B$. This then simplifies to
$$
\|B\|_1 \leq |\mathcal{O}| + 1 - \sum_{\omega' \in \mathcal{O}} \frac{\pi(\omega')}{\sigma(\omega')}, 
$$
completing the proof. 
\end{proof}

It is standard to define a ``complete'' binary tree to be one where every level except possibly the last is completely filled, and moreover, any leaves on the last level are placed as far left as possible. Also, a ``perfect'' binary tree refers to a binary tree where all levels are completely full.  Our next proposition demonstrates that it is not at all unusual for the $1$-norm of $B$ to be bounded by a constant.

\begin{prop}
Let $n := |\Omega| \geq 2$. If $\sigma(\omega) = 1/n$ for all $\omega \in \Omega$, then 
$$\|B\|_1 = \frac{n}{2^{\left\lfloor \log_2(n) \right\rfloor}} \leq 2.$$ 
\end{prop}
\begin{proof}
Since $\sigma$ is uniform, without loss of generality, $T(C)$ is the unique complete binary tree with $n$ leaves. If $\log_2(n) \in \mathbb{Z}$, then $\pi(\omega) = \sigma(\omega)$ for all $\omega \in \Omega$. Then $B = I_n$, so $\|B\|_1 = 1$. 

Now suppose $n$ is not an integral power of two. In this case, $T(C)$ is not a perfect binary tree, so there exists a positive integer $h$ so that some leaves of $T(C)$ occur at depth $h$ and the rest occur at depth $h+1$. Thus, the states in $\mathcal{O}$ are exactly the ones represented by leaves in $T(C)$ at depth $h+1$, so then $\mathcal{U}$ contains the states of $\Omega$ represented by leaves in $T(C)$ at depth $h$. Additionally, recall that a perfect binary tree with leaves at depth $h'$ contains $2^{h'}$ leaves. Thus, $h$ satisfies $2^h < n < 2^{h+1}$. Therefore, $h = \left\lfloor \log_2(n) \right\rfloor$, so for each $\omega \in \mathcal{U}$, $\pi(\omega) = 1 / 2^{\left\lfloor \log_2(n) \right\rfloor}$. 

We know $\|B\|_1$ is achieved by a column sum of $B$ corresponding to a state of $\mathcal{U}$ since the diagonal entry of each column corresponding to a state in $\mathcal{U}$ is $1$, whereas columns corresponding to states of $\mathcal{O}$ have a value at most one on the diagonal and zeros elsewhere. Since $b_{\omega\omega'}$ denotes the probability that state $\omega$ is transformed to state $\omega'$ by the ENCORE algorithm, $\pi(\omega') = \sum_{\omega \in \Omega} \sigma(\omega) b_{\omega\omega'}$ for each $\omega \in \Omega$. Thus, for any $\omega' \in \mathcal{U}$,
$$
\frac{1}{2^{\left\lfloor \log_2(n) \right\rfloor}} = \pi(\omega') = \sum_{\omega \in \Omega} \sigma(\omega) b_{\omega\omega'} = \frac{1}{n} \left(\sum_{\omega \in \Omega} b_{\omega\omega'}\right), 
$$
so 
$$
\|B\|_1 = \frac{n}{2^{\left\lfloor \log_2(n) \right\rfloor}} \leq \frac{n}{2^{\log_2(n) - 1}} = \frac{n}{n/2} = 2. 
$$
\end{proof}

Next, we bound the number of bits of entropy needed for ENCORE to perform its transformations. Define $H(p) := -\sum_{\omega \in \Omega} p(\omega)\log(p(\omega))$ to be the \textit{entropy} of distribution $p$. For two distributions $p$ and $q$ on the same finite state space $\Omega$, define the \textit{cross-entropy} of $q$ relative to $p$ as $H(p, q) = -\sum_{\omega\in\Omega} p(\omega)\log(q(\omega))$. 

\begin{lemma}\label{Lem:ShannonBound} \cite{Shannon1948}
Let $p$ be a distribution on the states of $\Omega$, and define $l(\omega) = -\log(p(\omega))$ for any $\omega\in\Omega$. Suppose $l'$ is the length function of some prefix-free encoding of $\Omega$.  If $H(\Omega)$ denotes the entropy of $\Omega$, then 
$$
H(p) \leq \mathbb{E}(l(\omega)) < H(p) + 1, 
$$
with equality if and only if $p$ is dyadic. Moreover, if $p$ is dyadic, 
$$
\mathbb{E}(l(\omega)) \leq \mathbb{E}(l'(\omega)) \mbox{ for all } l'.
$$
Finally, 
$$
\mathbb{E}(l(\omega)) \leq \mathbb{E}(l'(\omega) + 1) \mbox{ for all } l'.
$$
for any $p$. 
\end{lemma}

\begin{definition}
For discrete probability distributions $p$ and $q$ defined on the same sample space $\Omega$, the (base $x$) relative entropy from $q$ to $p$ is defined as
$$
\KL_x(p\|q) = \sum_{\omega\in\Omega} p(\omega) \log_x\left(\frac{p(\omega)}{q(\omega)}\right), 
$$
and is called the Kullback-Leibler (KL) divergence. When $x = 2$, the subscript is suppressed. 
\end{definition}

\begin{lemma}[Pinsker's Inequality \cite{CoverThomas}] \label{Lem:Pinsker}
If $p$ and $q$ are two probability distributions on a measure space $(X, \Sigma)$, then 
$$
\|p - q\| \leq \sqrt{\frac{1}{2} \KL_e(p\|q)}
$$
\end{lemma}

The following theorem shows that not many more bits are needed (ideally) to communicate the distribution $\sigma$ than are already needed to communicate $\pi$, so the transformation reconstruction data is well-controlled.

\begin{theorem}\label{Thm:BoundedWaste}
Let $\sigma$ be a distribution over the finite set $\Omega$, and let $\pi$ denote the probability distribution over $\Omega$ implied by Huffman coding. If $M$ denotes the length of the longest codeword defined by $\pi$, i.e., $M = \max_{\omega\in \Omega}\{-\log(\pi(\omega))\}$, then 
$$
\left|H(\sigma, \pi)-H(\pi) \right| \leq M\sqrt{2\ln(2)} < M.
$$
\end{theorem}
\begin{proof}
We begin with the following computation: 
\begin{align*}
|H(\sigma, \pi) - H(\pi)| &= \left|\left(-\sum_{\omega\in\Omega} \sigma(\omega) \log(\pi(\omega))\right) - \left(-\sum_{\omega\in\Omega} \pi(\omega) \log(\pi(\omega))\right)\right| \\
&= \left| \sum_{\omega\in\Omega} (\pi(\omega) - \sigma(\omega)) \log(\pi(\omega)) \right| \\ 
&= \left| \sum_{\omega\in\Omega} (\sigma(\omega) - \pi(\omega)) \log(1 / \pi(\omega)) \right| \\ 
&\leq M \sum_{\omega\in\Omega} \left| \sigma(\omega) - \pi(\omega) \right| \\ 
&\leq 2M \|\sigma - \pi\|
\end{align*}
Lemma \ref{Lem:Pinsker} states the following bound on $\|\sigma - \pi\|$. 
$$
\|\sigma - \pi\| \leq \sqrt{\frac{1}{2} \KL_e(\sigma\|\pi)}
$$
Notice that 
$$
\KL(\sigma\|\pi) = \sum_{\omega\in\Omega} \sigma(\omega) \log\left(\frac{\sigma(\omega)}{\pi(\omega)}\right) = \frac{1}{\ln(2)} \sum_{\omega\in\Omega} \sigma(\omega) \ln\left(\frac{\sigma(\omega)}{\pi(\omega)}\right) = \frac{1}{\ln(2)} \KL_e(\sigma\|\pi).
$$
Thus, 
$$
\|\sigma - \pi\| \leq \sqrt{\frac{\KL(\sigma\|\pi) \ln(2)}{2}}. 
$$
Let $l(\omega) = \left\lceil-\log(\sigma(\omega)) \right\rceil$, and define $l'(\omega)$ to be the length of the codeword corresponding to $\omega$ in a Huffman encoding. By Lemma \ref{Lem:ShannonBound}, $\mathbb{E}(l(\omega)) < H(\sigma) + 1$. Since $l'$ denotes codeword lengths in a Huffman scheme, the optimality of Huffman encoding implies $\mathbb{E}(l'(\omega)) \leq \mathbb{E}(l(\omega))$. Furthermore, by the Shannon source coding theorem (Theorem 5.4.3 in \cite{CoverThomas}), $H(\sigma) \leq \mathbb{E}(l'(\omega))$. Therefore, we have 
$$
H(\sigma) \leq \mathbb{E}(l'(\omega)) \leq \mathbb{E}(l(\omega)) < H(\sigma) + 1.
$$
Since $H(\sigma, \pi) = \mathbb{E}(l'(\omega))$, $|H(\sigma, \pi)-H(\sigma)| \leq 1$. Then $\KL(\sigma\|\pi) = |H(\sigma, \pi)-H(\sigma)| \leq 1$. Therefore, 
$$
\|\sigma - \pi\| \leq \sqrt{\frac{\KL(\sigma\|\pi) \ln(2)}{2}} \leq \sqrt{\frac{\ln(2)}{2}}, 
$$
so 
$$
|H(\sigma, \pi) - H(\pi)| \leq 2M \sqrt{\frac{\ln(2)}{2}} = M \sqrt{2 \ln(2)}, 
$$
completing the proof. 
\end{proof}

\section{Future Directions}\label{Section:Questions}

Here, we collect several open questions about the performance and cryptographic properties of the ENCORE algorithm.  We believe several of the statements of results above are weaker than what is in fact true, both in general and under stronger assumptions that the encoded stream is sampled from {\em typical} real-world data, i.e., distributions like those sampled from a Markov chain.  First, it may be the case that the bound in Theorem \ref{Thm:BoundedWaste} is unnecessarily weak, so the extra bits needed to encode $\sigma$ beyond its entropy rate is in truth smaller than what is stated.

\begin{question}
    Is it possible to improve the bound of $M\sqrt{2 \ln 2}$ in Theorem \ref{Thm:BoundedWaste}?
\end{question}

Above, we made significant use of the matrix $B$ which encodes transformation probabilities from the original samples to the distribution implied by Huffman coding.  In particular, $\|B\|_1$ impacts the mixing time of our transformed stream and therefore also the duration between communications obtained by an adversary and bits which are still unpredictable conditional on this information.  Thus, showing that $\|B\|_1$ is not large is important -- and we suspect that some of the bounds above can be strengthened.  In particular, the bound $\|B\|_1 \leq |\Omega|$ which is used in the proof of Lemma \ref{Lem:one-normBbound} seems very crude.

\begin{question}
    Is it possible to prove stronger bounds on $\|B\|_1$?
\end{question}

Next, one of the useful properties of the protocol we analyze above is that, in practice, it requires many fewer bits of entropy to execute than ordinary cryptographic protocols typically need.  Since truly random bits are famously expensive to obtain, it would be helpful to bound the number of bits that are required by the ENCORE algorithm.  Recall that matrix $B$ of size $|\Omega|\times |\Omega|$ has entry $b_{\omega\omega'}$ defined to be the probability of transforming state $\omega$ to state $\omega'$. We also write $\mathcal{O}$ for the set of overrepresented states, i.e., the collection of states where their corresponding row in $B$ has some nonzero off-diagonal entry.  Furthermore, recall that $B$ is row-stochastic, so for any $\omega \in \Omega$, $\sum_{\omega'\in\Omega} b_{\omega\omega'} = 1$, so we may refer to the distribution $\mathbf{B}_\omega$ represented by the $\omega$ row of $B$.  Then the expected number of bits of entropy generated in sampling from $\mathbf{B}_\omega$ is given by
\begin{align*}
\mathbb{E}_{\omega \sim \sigma}[H(\mathbf{B}_\omega)] &= \sum_{\omega\in\Omega} \sigma(\omega) \left( -\sum_{\omega'\in\Omega} b_{\omega\omega'} \log(b_{\omega \omega'}) \right) \\
&= \sum_{\omega\in O} \sigma(\omega) \left( -\sum_{\omega'\in\Omega} b_{\omega\omega'} \log(b_{\omega \omega'}) \right)
\end{align*}
Thus, if $M'_\omega = \max_{\omega'\in\Omega}\{-\log(b_{\omega\omega'})\}$, then we have
\begin{align*}
\mathbb{E}_{\omega \sim \sigma}[H(\mathbf{B}_\omega)] & \leq \sum_{\omega\in O} \sigma(\omega) M'_\omega \left( \sum_{\omega'\in\Omega} b_{\omega,\omega'} \right) \\
&= \sum_{\omega\in O} \sigma(\omega) M'_\omega .
\end{align*}
Lastly, if $M' = \max_{\omega\in\Omega}\{M'_\omega\}$, i.e., the negative log of the smallest nonzero entry of $B$, then 
$$
\mathbb{E}_{\omega \sim \sigma}[H(\mathbf{B}_\omega)] \leq M' \sum_{\omega\in O} \sigma(\omega) \leq M'.
$$
However, we strongly suspect that the above calculation gives away far too much in the inequalities, i.e., a much tighter bound is possible.  Experimentally, when $(\Omega,\sigma)$ consists of sequences of samples drawn from a large Markov chain, the number of states with a given degree of overrepresentation is similar to the number of states with a comparable degree of underrepresentation, so $\mathbf{B}_\omega$ almost always is supported on $O(1)$ states.  This is reasonable to expect, given the so-called ``Asymptotic Equipartition Property'' (see \cite{CoverThomas}) and the fact that stationary distribution probabilities are no more likely to be above or below nearby dyadic rationals.  Therefore, we ask:

\begin{question}
    Is it true that $\mathbb{E}_{\omega \sim \sigma}[H(\mathbf{B}_\omega)] \ll M'$?  Perhaps it is even $O(1)$ for ``most'' distributions $\sigma$ of state sequences sampled from a Markov chain as the number of states tends to infinity?
\end{question}

Finally, we suspect that the conclusion of Proposition \ref{Prop:MarkovChainBitProb} -- and therefore many of the following results -- can be strengthened considerably in the case that the length of codewords is concentrated around its mean.  This is again something to be expected in practice due to the Asymptotic Equipartition Property, and so we ask:

\begin{question}
    Is it possible to strengthen the bounds on deviation from uniform bit strings for distributions $\sigma$ so that $\Pr_{x \sim \sigma}[x]$ is concentrated around $\mathbb{E}_{x \sim \sigma}[\Pr[x]]$?
\end{question}

\end{document}